\documentclass[a4paper,onecolumn,accepted=2026-04-08,10pt]{quantumarticle}
\pdfoutput=1

\usepackage[utf8]{inputenc}
\usepackage[T1]{fontenc}
\usepackage[english]{babel}

\usepackage{amsmath,amssymb,amsfonts,amsthm,mathtools,braket}

\usepackage[backend=bibtex,style=phys,biblabel=brackets,doi=false,
  url=true]{biblatex}
\renewbibmacro*{doi+eprint+url}{%
  \printfield{doi}%
  \newunit\newblock
  \iftoggle{bbx:eprint}{\usebibmacro{eprint}}{}%
  \newunit\newblock
  \iftoggle{bbx:url}{\usebibmacro{url+urldate}}{}%
}

\usepackage{xurl}
\usepackage[colorlinks=true,citecolor=blue,linkcolor=red,filecolor=blue,urlcolor=blue]{hyperref}

\usepackage{setspace}
\usepackage{graphicx}
\usepackage{float}
\usepackage{subcaption}
\usepackage{tikz}
\usepackage{tikz-cd}
\usepackage{neuralnetwork}
\usetikzlibrary{cd}

\usepackage{xcolor}
\usepackage{tcolorbox}
\usepackage{empheq}

\newtcolorbox{mybox}[3][]{
  colframe = #2!25,
  colback  = #2!10,
  coltitle = #2!20!black,
  title    = {#3},
  #1,
}

\newtheorem{theorem}{Theorem}[section]
\newtheorem{lemma}[theorem]{Lemma}

\theoremstyle{definition}

\newtheorem{example}[theorem]{Example}

\usepackage{dcolumn}

\date{}

\begin{document}

\title{Quantum recurrences and the arithmetic of Floquet dynamics}

\author[1,2]{Amit Anand}
\author[1]{Dinesh Valluri} 
\author[5]{Jack Davis}
\author[1,2,3,4]{Shohini Ghose}
\affil[1]{Institute for Quantum Computing, University of Waterloo, Waterloo, Ontario, Canada N2L 3G1}
\affil[2]{Department of Physics and Astronomy, University of Waterloo, Waterloo, Ontario, Canada N2L 3G1}
\affil[3]{Department of Physics and Computer Science, Wilfrid Laurier University, Waterloo, Ontario, Canada N2L 3C5}
\affil[4]{Perimeter Institute for Theoretical Physics, 31 Caroline St N, Waterloo, Ontario, Canada N2L 2Y5}
\affil[5]{QAT team, DIENS, \'Ecole Normale Sup\'erieure, PSL University, CNRS, INRIA, 45 rue d'Ulm, Paris 75005, France}

\maketitle

\begin{abstract}
The Poincaré recurrence theorem shows that conservative systems in a bounded region of phase space eventually return arbitrarily close to their initial state after a finite amount of time. 
An analogous behavior occurs in certain quantum systems where quantum states can recur after sufficiently long unitary evolution, a phenomenon known as quantum recurrence. Periodically driven (i.e.\ Floquet) quantum systems in particular exhibit complex dynamics even in small dimensions, motivating the study of how interactions and Hamiltonian structure affect recurrence behavior. While most studies treat recurrence in an approximate, distance-based sense, here we address the problem of \textit{exact}, state-independent recurrences in a broad class of finite-dimensional Floquet systems, spanning both integrable and non-integrable models. 
Leveraging techniques from algebraic field theory, we construct an arithmetic framework that identifies all possible recurrence times by analyzing the cyclotomic structure of the Floquet unitary's spectrum. 
This computationally tractable approach yields both positive results, enumerating all candidate recurrence times, and definitive negative results, rigorously ruling out candidate recurrence times for a given set of Hamiltonian parameters. 
We further prove that rational Hamiltonian parameters do not, in general, guarantee exact recurrences, revealing a subtle interplay between system parameters and long-time dynamics. 
Our findings sharpen the theoretical understanding of quantum recurrences, clarify their relationship to quantum chaos, and highlight parameter regimes of special interest for quantum metrology and control.

\end{abstract}

\section{Introduction}\label{sec:introduction}

The idea that a system's classical dynamics may eventually repeat itself has existed for a long time. In the late 19th century Henri Poincaré established the now-famous Poincaré recurrence theorem, which states that any classical conservative system evolving within a bounded region of phase space will, after a sufficiently long but finite amount of time, return arbitrarily close to its initial state \cite{poincar_1890_avantpropos}.  
The theorem is non-constructive: it does not provide a method for calculating the recurrence time, and so finding such times must be done case by case \cite{gimeno_upper_2017,suskind_complexity_2018, alvaro_recurrence_2023, rauer_recurrence_science_2018,freedman_quantumrecurrent_2024,levesque_rec_time_2025,ropotenko_the_2025,kotowski_tightboundsrecurrencetime_2026}.  
This question of repeated dynamics naturally arose in quantum physics as well, with an early result by Bocchieri and Loinger demonstrating that a quantum state vector evolving unitarily after a sufficiently long time must return arbitrarily close in Hilbert-Schmidt distance to its initial state \cite{bocchieri_1957_quantum}.
Interestingly, in the case of periodically driven quantum systems, the quantum recurrence time is strongly influenced by the integrability of the system \cite{venuti_2015_the}. Understanding recurrences in such systems can thus yield important insights into the structure of quantum dynamics and the mechanisms behind recurrences.

Periodically driven quantum systems have come under intense study in recent years and have been related to thermalization in closed systems \cite{kaufman_2016_quantum}, quantum many-body scars \cite{Pai_and_Pretko_scar_2019,Sen_scar_2020,Norio_scars_2020}, and exotic phases of matter such as discrete time crystals and Floquet topological phases \cite{roy_2017_floquet,sacha_2017_time,khemani_2019_a}. They also play a central role in the study of quantum chaos, displaying rich dynamics such as integrability-to-chaos transitions and dynamical Anderson localization \cite{Izrailev_Shepelyanskii_1980, fishman_rotor_anderson_1982}. 
Such systems are analyzed using Floquet theory~\cite{gfloquet_1883}, which describes the evolution of a quantum system with a time-periodic Hamiltonian, $H(t+T)=H(t)$, via the repeated application of a unitary operator. Specifically, the evolution over one period is captured by the Floquet operator $
U = \mathcal{T} \exp\left(-i \int_t^{t+T} H(t') \, dt'\right)$,
where $\mathcal{T}$ denotes time-ordering, and $T$ is the period of the external drive~\cite{grifoni_driven_1998}.

The question of quantum recurrences in Floquet systems was first addressed in \cite{hogg_1982_recurrence}, which showed that under certain conditions, the dynamics of a time-dependent Hamiltonian can exhibit quasi-periodicity. Subsequent work linked the presence of recurrences to the rationality of the parameters in a Floquet Hamiltonian \cite{fishman_rotor_anderson_1982}, and recent studies have refined this understanding. For example, \cite{pandit_2022_bounds} explores the algebraic constraints on the recurrence of observables in Floquet systems, while \cite{oszmaniec_2024_saturation} establishes a relationship between recurrence times and the circuit complexity of the corresponding unitary evolution.

In many previous studies, the notion of recurrence has typically relied on an approximate, distance-based definition, rather than the exact recurrence exhibited in certain Floquet systems~\cite{khemani_2019_a, anand_davis_ghose_2024}. The recurrences explored in this work are \textit{state-independent} and \textit{exact}, in contrast to the \textit{state-dependent} and \textit{approximate} recurrences guaranteed by the quantum Poincaré recurrence theorem~\cite{bocchieri_1957_quantum}. Previous works have shown that exact state-independent recurrences are typically robust to small perturbations in the Hamiltonian parameters: they are often accompanied by nearby state-independent approximate recurrences~\cite{sayan_time_crystal_2020,Zou_Wang_pseudo_2022, anand_davis_ghose_2024}. Moreover, such exact recurrences provide sharply defined recurrence times, which have recently been shown to be essential for certain metrological protocols~\cite{zou_2025_enhancing,biswas_2025_the}. These protocols rely on the evolution reaching specific recurrence times, allowing one to ignore the intermediate dynamics entirely. In contrast, recurrence times associated with approximate, state-dependent recurrences in non-integrable systems tend to be extremely long, rendering them unfeasible for metrological applications. There has consequently been a growing interest in dynamical regimes that exhibit such exact recurrences, particularly in Floquet systems, where they can be exploited to design novel metrological protocols offering a quantum advantage.

In this work, we provide a general tool to identify the existence of exact state-independent recurrences in Floquet quantum systems across a broad class of Hamiltonians and Hilbert space dimensions, encompassing both integrable and non-integrable models. We leverage techniques from algebraic field theory to develop our framework for identifying such recurrences in finite-dimensional quantum systems.
While arithmetic methods have previously been employed to study quantum chaos in dynamical systems \cite{jens_aqc_1993, bogomolny_arithmetical_1997, marklof_arithmetic_2006}, our focus is fundamentally different: we ask whether exact recurrences in Floquet systems can be efficiently determined and, crucially, when they can be ruled out. One should note that if all the eigenvalues of a unitary operator $\{e^{i\theta_k}\}_k$ are known explicitly, the problem is trivial.  In particular, a unitary exhibits exact state-independent recurrence if and only if each phase angle $\theta_k$ is a rational multiple of $\pi$. However, for a general unitary matrix, it is not possible to determine all eigenvalues analytically for dimension greater than four.

To address this, we apply the theory of cyclotomic field extensions to the splitting field \cite{dummit_foote_abstract_2003} of the characteristic polynomial of the Floquet unitary. 
The degree of such an extension is intimately connected to the system’s potential for temporal periodicity. Critically, this does not require explicit knowledge of the eigenvalues -- it suffices to identify the number field to which the unitary matrix elements belong.
Using this approach, we rigorously determine, for a given Floquet unitary matrix $U$ with algebraic entries (a case that includes many physically relevant models), all possible values $n \in \mathbb{N}$ for which $U^n = \tau I$ for some residual phase $\tau = e^{i\theta} \in \mathrm{U}(1)$. The key strength of our method lies in its \textit{finiteness}: the set of candidates for $n$ is finite and checkable. If none of the candidates satisfy the recurrence condition, the existence of an exact recurrence for that system can be conclusively ruled out.  We furthermore derive an upper bound on the possible values of $n$ in terms of the system's Hamiltonian parameters and dimension, whereas previous work on recurrence times gives approximate scalings depending only on system size \cite{gimeno_upper_2017,suskind_complexity_2018, freedman_quantumrecurrent_2024,levesque_rec_time_2025,ropotenko_the_2025}.
This provides a powerful diagnostic tool for probing recurrence structure in driven quantum systems.

To illustrate our method, we study the quantum kicked top (QKT) model, a well-studied Floquet system which has a chaotic classical limit \cite{haake_1987}. It is a finite-dimensional, periodically driven spin system with conserved total angular momentum $j$. Being a single spin-$j$ system, it is furthermore equivalent to a system of indistinguishable qubits, which allows techniques from entanglement theory to be applied, in particular for ruling out the candidate recurrence times.  In previous works, state dependent temporal periodicity of physical observables for small $j$ has been studied  \cite{Ruebeck_Pattanayak_2017,Bhosale_Santhanam_periodicity_2018,Dogra_exactly_2019,bhosle_qkt__jan_2024,bhosle_nov_2024,bhosle_qkt_aug_2024}. In our recent work \cite{anand_davis_ghose_2024} we studied quantum recurrences in the quantum kicked top in more generality, and demonstrated an infinite family of QKT dynamics with purely quantum recurrences that do not appear in the classical limit of a chaotic system. These recurrences are state-independent and thus do not correspond to classical periodic orbits, which are contingent on the initial state. We analytically proved the existence of these recurrences for certain sets of Hamiltonian parameters across all finite dimensions. However, in the case of a negative result, we were only able to provide a numerical verification.
Extending our previous work, we now employ the above method to systematically identify exact recurrences for a significantly larger set of Hamiltonian parameters. Crucially, the framework is not only capable of finding all such recurrences when they exist, but also of delivering a definitive negative result—rigorously ruling out the possibility of any exact recurrence for the given parameters. In the latter case, verification reduces to a finite, parameter-dependent search, which can be carried out numerically with computational resources commensurate to the system size.

The paper is structured as follows: in Section \ref{sec: method}, we introduce the arithmetic framework used to find possible recurrence times based on the Hamiltonian parameters. We determine the set of possible recurrences given a set of Hamiltonian parameters. Section \ref{sec: QKT} applies this approach to a specific example—the quantum kicked top with spin-$3/2$. Here, we identify both the parameter sets that lead to exact periodic behavior and those that do not. Section  \ref{sec: discussion} summarizes the broader implications of our results and highlights how arithmetic methods can provide insights into recurrence behavior in driven quantum systems. The key mathematical background—especially on algebraic field extensions—is presented in Appendix \ref{sec: appendix}.

\section{Method}\label{sec: method}

In this section, we present our method to study the exact recurrence of quantum systems evolved under unitary Floquet dynamics. Let the dynamics of the Floquet system be governed by the time-periodic Hamiltonian, $H(t)$ satisfying
\begin{equation}
    H(t) = H(t+T),
\end{equation}
where $T$ is the time-period.  The corresponding unitary evolution operator for one time period is $U$, as defined previously (see introduction section \ref{sec:introduction}).  We say that $U$ is \textit{quasi-periodic} or simply \textit{periodic} by abuse of language, if there exists an $n \in \mathbb{N}$ such that $U^n = \tau I$ for some $\tau = e^{i \theta} \in \mathrm{U}(1)$.  We call the smallest such positive integer $n$ the \textit{period} of $U$.
Then the steps for finding such $n$ are as follows:
\begin{itemize}
    \item Let $U$ be a $d\times d$ unitary matrix.  We assume there exists a $n$ such that $U^n= \tau \mathbb{I}$, where $\tau=e^{i\theta} \in \mathrm{U}(1)$ is some global phase with phase angle $\theta$.
    \item In Lemma \ref{proof:lemma 1} we establish the consequential restrictions on the eigenvalues $\{\lambda_{i}\}_{i=1}^{d}$ of $U$.  Namely, we are able to construct a primitive $n$-th root of unity $\zeta_n$ in terms of the eigenvalues.  
    \item Next we provide a theorem used to calculate a finite set of all such possible $n$, using the fact that $\zeta_n$ belongs to the splitting field of the characteristic polynomial $p_{U}(t) := \text{det}(U-tI)$ of $U$. This also gives an upper bound on $n$ that only depends on $d$ and the degree of the field over which $U$ is defined.
    \item Once we get the finite set of candidate $n$ values, we calculate physically relevant quantities such as von Neumann entropy to test the existence or non-existence of exact recurrences.
    \item In the case of vanishing von Neumann entropy, we need to check the action of $U^n$ on the set of vectors forming a complete basis to confirm the exact recurrences. 
\end{itemize}

Observe that if $\{\lambda_{i}\}_{i=1}^{d}$ are eigenvalues of the matrix $U$ with period $n$, then for every index $i$ we have 
\begin{equation*}
    \lambda_{i}^{n} = \tau    
\end{equation*}
for some fixed phase $\theta$ of the form $\tau = e^{i\theta}$ depending only on $U$. In particular, $\lambda_{i} = \zeta_{n}^{k_i} \tau^{1/n},$ where $\zeta_n$ is a primitive $n^{\text{th}}$ root of unity and $k_{i}$ are integers. The fact that $n$ is the period of $U$ allows us to express $\zeta_n$ purely in terms of $\lambda_1, \ldots, \lambda_d$ and also restricts the possible values of each $k_i$. To achieve this, we need the following lemma.
\begin{lemma} \label{lemma:lemma 1}
    Suppose $U$ is a $d\times d$ periodic matrix with period $n$ and $\lambda_{i} = \zeta_{n}^{k_i} \tau^{1/n}$ are the eigenvalues of $U$ as described above. Then the following are true: 
    \begin{itemize}
        \item[(a)] $R := gcd(k_1 - k_d, \ldots, k_{d-1} - k_d)$ is invertible modulo $n$. In other words, $gcd(R,n) = 1$.
        \item[(b)] There exist integers $a_1, \ldots a_d$ such that $a_1 + \ldots + a_d = 0$ and $a_1 k_1 + \ldots + a_d k_d = 1 \pmod{n}$. 
    \end{itemize}
\end{lemma}
\begin{proof}\label{proof:lemma 1}
     (a.) Suppose $R = gcd(k_1 - k_d, \ldots, k_{d-1} - k_d)$, in particular $k_i \equiv k_d \pmod{R}$ for all $1 \leq i \leq d-1.$  write $k_i = k_{d} + Rl_{i}$ for some integers $l_i,$ for all $1 \leq i \leq d-1.$ We need to prove that $gcd(R,n) =1.$
     Observe,
     $\zeta_{n}^{k_i} = \zeta_{n}^{k_d} \zeta_{n}^{Rl_i} $. Since $\lambda_{i} = \zeta_{n}^{k_i} \tau^{1/n}$, we have 
     \begin{align}\label{eq: ratio of two eigne value}
         \frac{\lambda_i}{\lambda_d} &= \frac{\zeta_{n}^{k_i}}{ \zeta_{n}^{k_d}} = \zeta_{n}^{Rl_i}.
     \end{align}
     By raising both sides of Eq. \ref{eq: ratio of two eigne value} to the power of $n/gcd(R,n)$ we get
     \begin{align}
          \bigg(\frac{\lambda_i}{\lambda_d}\bigg)^{n/gcd(R,n)} = \bigg(\zeta_{n}^n\bigg)^{l_i\frac{R}{gcd(R,n)}} = 1,
     \end{align}
    since $\zeta_{n}$ is a $n$-th root of unity and  $Rl_i/gcd(R,n)$ is an integer. It follows that
\begin{equation}\label{eq:diag-power}
\begin{pmatrix}
\lambda_{1} &        &        &        \\
            & \ddots &        &        \\
            &        & \ddots &        \\
            &        &        & \lambda_{d}
\end{pmatrix}^{\tfrac{n}{\gcd(R,n)}}
\;=\;
\begin{pmatrix}
\lambda_{d} &        &        &        \\
            & \ddots &        &        \\
            &        & \ddots &        \\
            &        &        & \lambda_{d}
\end{pmatrix}^{\tfrac{n}{\gcd(R,n)}}
\;=\;
\lambda_{d}^{\tfrac{n}{\gcd(R,n)}}\,I.
\end{equation}
This implies $U^{n/gcd(R,n)} = \tau' I $ for $\tau' = \lambda_d^{n/gcd(R,n)} \in \mathrm{U}(1)$. This is only possible when $gcd(R,n) = 1,$ as $U$ has period $n$.

(b.) By Bezout's identity there exists integers $a_1, \ldots, a_{d-1}$ such that 
\begin{equation}\label{eq: writing Bezout identity for R}
    R := gcd(k_1 - k_d, \ldots, k_{d-1} - k_d) = b_1(k_1 - k_{d}) + \ldots + b_{d-1} (k_{d-1} - k_d) 
\end{equation}
for some integer $b_i$'s.  Define $b_d := -(b_1 +\ldots + b_{d-1}),$ and let $S$ be the inverse of $R$ modulo $n$, i.e., $RS = 1 \pmod{n}$. By multiplying $S$ on both sides of Eq. \eqref{eq: writing Bezout identity for R}, we get $Sb_1(k_1 - k_{d}) + \ldots + Sb_{d-1} (k_{d-1} - k_d) = RS = 1 \pmod{n}$. By defining $a_i :=Sb_i$ for $1\leq i \leq d$, we get 
\begin{align}
    a_1(k_1 - k_{d}) + \ldots + a_{d-1} (k_{d-1} - k_d) &= 1 \pmod{n},  \, \,  \text{since}\\
     \nonumber a_1(k_1 - k_{d}) + \ldots + a_{d-1} (k_{d-1} - k_d) & = a_1k_1 + \ldots + a_{d-1}k_{d-1} - ( a_1 + \dots + a_{d-1})k_d \\
   \nonumber  & = a_1k_1 + \ldots + a_dk_d = 1 \pmod{n},
\end{align}

as required in the lemma. 
\end{proof}

So far, we have related the eigenvalues of a periodic unitary operator $U$ with its hypothetical period $n$. Now, we will give a procedure to find $n$ if it exists. To proceed, we will recall the notions of a field, a field extension, and the \textit{splitting field} of a polynomial. Informally, a field is a set in which addition, (commutative) multiplication, subtraction, and division by any non-zero element are possible. If $L$ and $K$ are two fields, and $K \subset L$ then we say $L$ is a \textit{field extension} (or simply an extension) of $K$ or equivalently we say that $K$ is a \textit{sub-field} of $L$. For our purpose, we only ever need sub-fields of $\mathbb{C}$, the field of complex numbers. Now suppose $p(t)$ is a polynomial with coefficients in a field $K$ and $\lambda_1, \ldots ,\lambda_d$ are its complex roots. The splitting field of $p(t)$ is defined as the field extension $L:=K(\lambda_1,\dots,\lambda_d)$, the smallest field containing $K$ and all the roots of $p(t)$.  
If $K \subset L$ is a field extension, then $L$ is naturally a vector space over the field $K$. Therefore, we can speak of the dimension $dim_{K}(L)$ of $L$ over $K$, which can either be finite or infinite.  It is usually denoted by $[L:K]$ and called the \textit{degree} of $L$ over $K$. If the degree $[L:K]$ is finite, we say that $L$ is a finite extension of $K$. If $K \subset M \subset L$, where $K\subset L$ is a finite extension then $K \subset M$ and $M \subset L$ are finite extensions, and we have the tower law 
\begin{equation}\label{eq: tower law}
    [L:K] = [L:M] [M:K].
\end{equation}

Now we will put the first restriction on our unitary $U$. We assume that all entries of $U$  are algebraic over $\mathbb{Q}$, the field of rational numbers. From now on, we define $K$ as the field $K:= \mathbb{Q}(\{u_{ij}\})$ where $u_{ij}$ are the entries of the unitary $U$.  
By definition, the field $K$ is the smallest field containing the entries of $U$. Since these entries are algebraic over $\mathbb{Q}$, the degree of the field extension $\mathbb{Q}\subset K$, denoted as $[K:\mathbb{Q}]$, is finite.

Recall that the characteristic polynomial of the matrix $U$ is given by  $p_{U}(t) = \text{det}(U-tI)$ and that the eigenvalues $\lambda_1, \ldots, \lambda_d$ are the complex roots of this polynomial. The coefficients of this polynomial are contained in $K$, or in other words, $p_{U}(t)$ is defined over $K$. The splitting field of $p_{U}(t)$ is defined as $L := K(\lambda_1, \ldots, \lambda_d)$. Using the fact that $\lambda_{i} = \zeta_{n}^{k_i} \tau^{1/n}$ and part (b) of Lemma \ref{lemma:lemma 1}, we get
\begin{align} \label{eq: zeta_n in terma of lambda}
  \nonumber  \lambda_1^{a_1} \ldots \lambda_d^{a_d} &= \bigg(\zeta_{n}^{k_1} \tau^{1/n}\bigg)^{a_1} \ldots \bigg(\zeta_{n}^{k_d} \tau^{1/n}\bigg)^{a_d} \\
 \nonumber   &=\zeta_{n}^{\sum a_i k_i} (\tau^{1/n})^{\sum_{i=1}^{d} a_i} \\
    &= \zeta_n.
\end{align}
Since $\zeta_n$ is a product of the eigenvalues, we conclude that it lies in the splitting field $L$.
Furthermore, using  $\lambda_{i} = \zeta_{n}^{k_i} \tau^{1/n}$ and Eq. \eqref{eq: zeta_n in terma of lambda},  $\tau^{1/n}$ can be expressed as ,
\begin{equation}\label{eq: tau in terms of lambdas}
    \tau^{1/n} = \frac{\lambda_i}{(\lambda_1^{a_1} \ldots \lambda_d^{a_d})^{k_i}}
\end{equation}
This shows that $\tau^{1/n}$ also belongs to the splitting field $L$. On the other hand, each $\lambda_{i}$ can be expressed in terms of $\zeta_n$ and $\tau^{1/n}$. Therefore, the field $K(\zeta_n, \tau^{1/n})$ contains all $\lambda_i's$. This implies that the splitting field $L$ defined above can also be expressed as $L = K(\zeta_n, \tau^{1/n}).$ In other words, we have established that the field containing $K$ and the eigenvalues of $U$ is the same as the field containing $K$, $\zeta_n$, and $\tau^{1/n}$. Using this relation, we give a theorem to calculate the allowable periods $n$ of $U$, which in turn depend on the dimension of the field $K$ over $\mathbb{Q}$. 
\begin{theorem}\label{Theorem 1}
    If $U$ is a $d \times d$ unitary matrix whose entries are in a field $K,$ which is a finite extension of $\mathbb{Q}$ and is periodic with period $n$ then $\phi(n)$ divides $d![K:\mathbb{Q}],$ where $\phi(\cdot)$ is the Euler totient function.
\end{theorem}
\begin{proof}\label{proof:Theorem 1}
    The characteristic polynomial $p_{U}(t)$ of $U$ has coefficients in the field $K$ and has degree $d$.  
    The dimension $[L:K]$ of the splitting field $L$ of $p_{U}(t)$ over $K$ divides $d!$ \cite{dummit_foote_abstract_2003}. Using the tower law \eqref{eq: tower law}, we have $[L:K] = \frac{[L:\mathbb{Q}]}{[K:\mathbb{Q}]}$. Therefore, $\frac{[L:\mathbb{Q}]}{[K:\mathbb{Q}]}$ divides $d!$, which implies that $[L:\mathbb{Q}]$ divides $[K:\mathbb{Q}]d!.$ 
    
    Now to prove the theorem, it is enough to show that $\phi(n)$ divides $[L:\mathbb{Q}]$. Recall that $\zeta_n \in L$ (as observed using Eq.\ \eqref{eq: zeta_n in terma of lambda}) implies that  $\mathbb{Q} \subset \mathbb{Q}(\zeta_n) \subset L$. It is a well-known fact that $[\mathbb{Q}(\zeta_n):\mathbb{Q}] = \phi(n)$; see Appendix \ref{A1:algebraic externsions} for more information. Again using the tower law, $[L:\mathbb{Q}] = [L:\mathbb{Q}(\zeta_n)][\mathbb{Q}(\zeta_n):\mathbb{Q}]$. This implies, $\phi(n)$ divides $[L:\mathbb{Q}]$. Therefore, $\phi(n)$ divides $[K:\mathbb{Q}]d!.$ 
\end{proof}

\section{Application to single particle system: the quantum kicked top}\label{sec: QKT}

The quantum kicked top (QKT) serves as a finite-dimensional dynamical framework for investigating quantum chaos, known for its compact phase space and parameterizable chaoticity structure \cite{haake_1987}. This time-dependent system is periodically driven and governed by the Hamiltonian
\begin{equation}\label{eq:QKT_Hamiltonian}
 H = \hbar\frac{\alpha J_{y}}{T} +  \hbar \frac{\kappa J_{z}^2}{2j} \sum_{n=-\infty}^{\infty} \delta (t-nT).
\end{equation}
Here $\{J_{x}, J_{y}, J_{z}\}$ denote the generators of angular momentum, satisfying the commutation relation $[J_i, J_k] = i\epsilon_{ikl} J_l$.  The eigenstates of the $z$-rotation generator, $J_z\ket{j,m} = m\ket{j,m}$, are called Dicke states, and serve as the basis in which we will perform many of the calculations to follow.  The model describes a spin of size $j$ undergoing precession about the $y$-axis, accompanied by impulsive, state-dependent twists around the $z$-axis, characterized by the chaoticity parameter $\kappa$. The time interval between these impulsive kicks is denoted by $T$, and $\alpha$ represents the extent of $y$-precession within one period. The total angular momentum, $j$, remains conserved throughout the dynamics. The classical kicked top can be obtained by computing the Heisenberg equations for the re-scaled angular momentum generators, $X_{i}=J_i/j$, satisfying $[X_i, X_k] = (1/j)\epsilon_{ikl} X_l$ and followed by the limit $j \to \infty$ \cite{haake_1987}. In the classical limit, the system shows mixed dynamics as $\kappa$ increases from 0, and transitions to globally chaotic behavior for $\kappa \gtrsim 4.4$ \cite{kumari_2019_quantumclassical}.

The Floquet time evolution operator for a single period is given by:
\begin{equation}\label{eq:floquet_unitary QKT}
    U = \exp\Big(-i\frac{\kappa }{2j}J_{z}^2\Big) \exp\Big(-i\alpha J_{y}\Big).
\end{equation}
Here the first part of the Eq.\ \eqref{eq:floquet_unitary QKT} represents a non-linear operation (twist) about the $z$-axis and the second part represents a rotation operation about the $y$-axis. The Floquet operator $U$ is here factored as $U := U_{\kappa} U_{\alpha}$ where $U_{\kappa} = \exp\Big(-i\frac{\kappa }{2j}J_{z}^2\Big)$ is diagonal in the $z$-axis/Dicke basis, with $U_{\kappa} = diag(\mu^{a_1}, \ldots, \mu^{a_d})$ where $\mu := e^{-i \kappa/8j}$ depending on a real parameter $\kappa$ and $a_{1}, \ldots, a_{d}$ are fixed integers which depend on $j$. The $y$-rotation matrix $U_{\alpha} = \exp(-i\alpha J_{y})$ belongs to $\mathrm{SU}(2j+1)$. In the Dicke basis, it is referred to as the \textit{Wigner D-matrix} \cite{Biedenharn_Louck_Carruthers_1984}. By taking $\kappa$ and $\alpha$ as a rational multiple of $\pi$, the unitary $U$ given in Eq.\eqref{eq:floquet_unitary QKT} satisfies the condition of being algebraic entries. Therefore, our method of studying the recurrences can be directly applied to the kicked top unitary. 

To test whether a candidate recurrence time is plausible, we leverage additional structure of general spin-$j$ systems \cite{Biedenharn_Louck_Carruthers_1984}. 
In particular, there is an isomorphism between the Hilbert space of a spin-$j$ system and the symmetric subspace of the Hilbert space of $N=2j \in \mathbb N$ qubits:
\begin{equation}
    \mathbb C^{2j+1} \simeq \mathrm{Sym}((\mathbb C^2)^{\otimes N}) \subset (\mathbb C^2)^{\otimes N}. 
\end{equation}
Physically, the many-qubit picture of spin can be seen as the highest-weight multiplet formed by coupling $N=2j$ spin-$\frac12$ particles.  Mathematically, this isomorphism can be seen by identifying the Dicke basis $\ket{j,m}$ with the symmetrized computational basis $\ket{D_N^{(w)}}$:
\begin{equation}\label{eq:dicke_qubits}
    \ket{j,m} \simeq \ket{D_N^{(w)}} := \frac{1}{\sqrt{\binom{N}{w}}} \sum_{\pi} \ket{\underbrace{1 \cdots 1}_{w} \, \underbrace{0\cdots 0}_{N-w} },
\end{equation}
where $w=j-m$ is the Hamming weight $w \in \{0,...,N\}$, the spin-up state $\ket{\frac12, \frac12}$ is by convention mapped to the computational zero state $\ket{0}$, and the sum runs over distinct permutations of the $N$ qubits.  For example, the Dicke state $\ket{\frac32, \frac12}$ is identified with the standard $W$ state $\frac{1}{\sqrt{3}}(\ket{100}+\ket{010}+\ket{001})$ of entanglement theory \cite{3qubits2ways_Dur_2000}.  Moreover, the collective spin observables $J_i$ become the tensor product representation of the $N$-qubit spin observables:
\begin{equation}\label{eq:spin-observables-from-qubits}
    J_i \simeq S_i := \sum_{k=1}^{N} \frac{\sigma^{(k)}_i}{2}, \qquad \sigma^{(k)}_i = I \otimes \cdots \otimes \underbrace{\sigma_i}_{\mathclap{k\text{-th tensor factor}}} \otimes \cdots \otimes I.  
\end{equation}
One can readily show from \eqref{eq:dicke_qubits} and \eqref{eq:spin-observables-from-qubits} that $S_z\ket{D_N^{(w)}} = (\frac{N}{2}-w)\ket{D_N^{(w)}}$ as expected, and that the Pauli commutation relations imply the set $\{S_i\}$ realizes the $\mathfrak{su}(2)$ algebra, $[S_i,S_j] = i\epsilon_{ijk} S_k$.   This algebra acts irreducibly in the symmetric subspace, and it can be proved that this isomorphism is an intertwiner of SU(2) irreps \cite{harrow2013church}.  This means the Floquet unitary associated with the quantum kicked top in the qubit picture is simply Eq.\ \eqref{eq:floquet_unitary QKT} with $\{S_i\}$ replacing $\{J_i\}$.

The key benefit of this correspondence is that it allows an interpretation of collective spin-$j$ properties in terms of the multi-partite entanglement structure in a many-body system \cite{Giraud_Tensor_PRL_2015}.  In particular, consider the set of spin coherent states,
\begin{equation}
    \ket{\theta,\phi} := e^{-i\theta(J_y \cos \phi - J_x \sin\phi)} \ket{j,j},
\end{equation}
where the unitary is the rotation that takes the North Pole to the point $(\theta,\phi) \in S^2$ in spherical coordinates.  In the qubit picture, Eq.\ \eqref{eq:spin-observables-from-qubits} implies that such a rotation amounts to the simultaneous rotation of each qubit to the same point on the Bloch sphere,
\begin{equation}
\begin{aligned}
    \ket{\theta,\phi} &\simeq \Big(e^{-i\frac{\theta}{2}(\sigma_y \cos \phi - \sigma_x \sin\phi)} \ket{0} \Big)^{\otimes N} \\
    &= \Big( \cos \frac{\theta}{2} \ket{0} + e^{i\phi} \sin \frac{\theta}{2} \ket{1} \Big)^{\otimes N}.
\end{aligned}
\end{equation}
Since the many-qubit state must be permutation-invariant, this establishes the equivalence between spin coherent states and the set of symmetric product states $\{\ket{\psi_{\text{prod}}}\}$; see \cite{Arecchi1972SCS} for the original discussion.  Hence, for a quick way to rule out a candidate periodicity $n$, we compute the single-qubit von Neumann entropy between the $2j$ qubits to test whether it vanishes exactly, provided we start with a symmetric product state. A non-vanishing entanglement entropy of the state $U^n\ket{\psi_{\text{prod}}}$ implies that $U^n\ket{\psi_{\text{prod}}} \neq e^{i\theta}\ket{\psi'_{\text{prod}}}$ for any other product state $\ket{\psi'_{\text{prod}}}$ or phase $e^{i\theta}$, which implies that $U^n \neq \tau I$ and so invalidates the candidate $n$.
The single-qubit entanglement entropy
\begin{equation}\label{eq:entanglement entropy}
S(\rho_1) = -\text{Tr}[\rho_1 \log\rho_1], 
\end{equation}
of any one of the reduced constituent qubits $\rho_1$ can be computed in terms of the global spin observables via \cite{Baguette_MM_reductions_2014}:
\begin{equation}
\label{e: one qubit density matrix}
    \rho_1 = \frac{1}{2}
    \begin{pmatrix}
    1 - \langle J_z \rangle & \langle J_- \rangle \\
    \langle J_+ \rangle & 1 + \langle J_z \rangle
    \end{pmatrix}, \qquad J_\pm = J_x \pm i J_y.
\end{equation} 
If however the von Neumann entropy does vanish, all we may conclude is that the product state was mapped to another product state, but not necessarily the same one.  In this case, we would then, for example, need to analytically apply the given unitary with the corresponding power, $U^n$, to a suitable basis in order to confirm or reject the exact recurrence.

\subsection{Case of spin 3/2}
To illustrate our framework, we study the quantum kicked top for spin $j=\frac{3}{2}$ and calculate the set of possible $n$ for which Eq.\eqref{eq:floquet_unitary QKT} is periodic or quasi-periodic. After determining the set of integers $n$ for which the Floquet unitary may be periodic, we must either confirm or rule out each $n$.  The Floquet unitary for the $j=3/2$ QKT with $\alpha=\pi/2$ (a commonly used choice) can be given by two $4 \times 4$ matrices:
\begin{equation}
U=U_{\kappa}U_{\alpha}, \qquad  U_{\kappa} = \begin{pmatrix} \mu^9 & 0 & 0 & 0 \\ 0 & \mu & 0 & 0 \\ 0 & 0 & \mu & 0 \\ 0 & 0 & 0 & \mu^9\end{pmatrix}, \qquad U_{\alpha} = \begin{pmatrix} a & -b & b & -a \\ b & -a & -a & b \\ b & a & -a & -b \\ a & b & b & a\end{pmatrix},
\end{equation}
where $a = \frac{1}{2\sqrt{2}}$, $b = \frac{1}{2} \sqrt{\frac{3}{2}}$, $\mu = e^{\frac{-ik}{12}}$, and $k \in \mathbb R$. Note that $\mu^9 = e^{\frac{-i3k}{4}}$ and that $\det(U_{\alpha}) = 4(a^2 + b^2)^2 = 4 (\frac{1}{2})^2 = 1$. Therefore, $\text{det}(U) = \mu^{20}$. The characteristic polynomial of $U$ is $p_{U}(t) = \text{det}(tI - U)$, which is formally a polynomial in two variables, $\mu$ and $t$. In fact, it can be explicitly computed as 
\begin{equation}\label{eq:polynomial for QKT unitary}
    p_{U}(t) = 4\mu^{20}(a^2 + b^2)^2 + 4a(\mu^{8} - 1)\mu^{11} t(a^2 + b^2) + 2a^2 (\mu^8 - 1)^2 \mu^2 t^2 - 2a(\mu^8 - 1 )\mu t^3 + t^4,
\end{equation}
and simplified to
\begin{equation}\label{eq: simplified polynomial for QKT unitary}
   p_{U}(t) = \mu^{20} + 2a(\mu^{8} - 1)\mu^{11} t + 2a^2 (\mu^8 - 1)^2 \mu^2t^2 - 2a(\mu^8 - 1 )\mu t^3 + t^4 
\end{equation}
upon substituting $a^2 + b^2 = \frac{1}{2}$.  The polynomial $p_{U}(t)$ can now be viewed as a polynomial in the variable $t$ of degree $4$ over the field $K := \mathbb{Q}(a,b,\mu) = \mathbb{Q}(\sqrt{2}, \sqrt{3}, \mu)$. By adjoining roots of the polynomial $p_{U}(t)$ to the field $K$, we create the splitting field $L$. In this case, it is given by $L=K(\lambda_1,\lambda_2,\lambda_3,\lambda_4)$, where $\lambda_i's$ are the eigenvalues of unitary $U$ given in Eq. \eqref{eq:floquet_unitary QKT}. This puts us in a position to apply Theorem \ref{Theorem 1}.
To this end, we now assume that $U$ is periodic (or quasi-periodic) with period $n$. Using Theorem \ref{Theorem 1} we get that $\phi(n)$ divides $[K:\mathbb{Q}]d! = [\mathbb{Q}(\sqrt{2}, \sqrt{3}, \mu) : \mathbb{Q}]4!.$ 

Here, $\kappa$ may be chosen as either a rational or an irrational multiple of $\pi$. When $\kappa$ is an irrational multiple of $\pi$, the parameter $\mu$ can be either algebraic or transcendental; however, in this case exact recurrences do not occur, and the system’s dynamics densely fill the accessible phase space. Consequently, it suffices to restrict attention to the case where $\kappa$ is a rational multiple of $\pi$.
Now if we consider $\kappa$ as a rational multiple of $\pi$ i.e., $\kappa=\frac{p}{q}\pi$, we define $m:=\frac{24q}{gcd(24,p)}$. This makes $\mu$ a primitive $m^{\text{th}}$ root of unity. Therefore, $[K:\mathbb{Q}] = l\, \phi(m)$, where $l\in{(1,2,4)}$ depending upon the value of $\mu$. Finally, we get that $\phi(n)$  must divide $24 \,l\,\phi(m)$. 

For a specific chaoticity value as an example, taking $\kappa=j\pi$ for $j=3/2$ gives us $m=16$. Therefore, we get $\mu = e^{-i\frac{2\pi}{16}} = \zeta_{16}$. The degree of $[K:\mathbb{Q}] = [\mathbb{Q}(\sqrt{2},\sqrt{3},\zeta_{16}):\mathbb{Q}]=2 \, \phi(16) = 16$. Therefore, $\phi(n)$ divides $24\times16$. Using this, we find the set of all possibles $n$ which here is found to be a set of size 141 with the largest element being 1680. Upon explicit calculations done in the Dicke basis, we find that the unitary given in Eq. \eqref{eq:floquet_unitary QKT} for $j=3/2$ is quasi-periodic with period $n=12$ with $\kappa=j\pi$ and $\alpha=\pi/2$.

For the case of $\kappa=\frac{j \pi}{2}$, we have $\mu = \zeta_{32} = e^{\frac{2\pi i}{32}}$. Here, $m = 32,$ which gives us $\phi(32) = 16$. Therefore,  $\phi(n)$ divides $[K:\mathbb{Q}]\,d!= 2\times16 \times 24 = 768.$ Since $n$ is upper bounded by $2\times \phi(n)^2$, the maximum value of $n$ for which the unitary can be periodic (or quasi-periodic)is $2\times 768^2$. In this case, the set of admissible $n$ contains 183 elements, with the largest value being 3570. We have found that for every such possible $n$, the unitary generates a non-zero amount of entanglement entropy, see Eq.  \eqref{eq:entanglement entropy}, between the virtual qubits, and therefore cannot be associated with a recurrence.  This rigorously shows that there is no $n$ for which the spin $j=\frac{3}{2}$ Floquet operator at $\kappa= j \pi/2$ and $\alpha=\pi/2$ is periodic (or quasi-periodic). Our results are consistent with \cite{anand_davis_ghose_2024}, where the recurrence for $\kappa=j\pi$ was rigorously proven, but the lack of recurrence at $\kappa=j\pi/2$ was only numerically suggested.
We note that this provides a clear and rigorous example showing that, even when Hamiltonian parameters are rational multiples of  $\pi$, exact quantum recurrences are not guaranteed.

\subsection{Extension to periodically driven many-body systems}\label{sec: applications many-body }

In this section, we show that our method can be applied to many-body Floquet systems with stepwise drive. Consider a family of $T$-periodic Floquet Hamiltonians given by
\begin{align}  
\label{eq:QKT_Hamiltonian}
 H(t) & = H_1(t), \,\, nT \leq t \leq t_1 \\
  \nonumber  & = H_2(t), \,\, t_1 \leq t \leq (n+1)T
\end{align}
where $T$ is the Floquet time period and $[H_1,H_2] \neq 0$. Such systems are modelled by a unitary matrix which is of the form $U = U_1 U_2$, where $U_i = exp(-\frac{i}{\hbar}\int H_i dt)$ are unitaries. 
Similar to the single-particle system, the main idea of taking the unitary as a product of two unitaries is that we can study the exact recurrences in the dynamics depending on two different parameters of $U_1$ and $U_2$ separately.
In our formalism, we have taken $U_1$ to depend on different $\kappa_i$ and $U_2$ to have algebraic entries over $\mathbb{Q}$. 
More generally, we may assume $U_1$ is diagonal by simply changing the basis to one that diagonalizes $U_1$. The eigenvalues ($\mu_i$), i.e., the diagonal entries of $U_1$, are of the form $e^{-i \kappa_i}$ for some real parameter $\kappa_i$. We assume that all such $\kappa_i$ are rational multiples of $\pi$, or equivalently, that all the eigenvalues of $U_1$ are roots of unity. In particular, this implies that all entries of $U_1$ are algebraic over $\mathbb{Q}$. Furthermore, we assume that $U_2$ is in $\mathrm{SU}(d)$, the group of special unitary $d \times d$ matrices whose entries are also algebraic over ${\mathbb{Q}}$ (see \ref{sec: appendix}). Multiple important many-body systems, such as central spin models \cite{biswas_2025_the} and Ising models \cite{bukov_many_body_2021} satisfy the above restrictions. Our method can hence be directly applied to such models for finding exact quantum recurrences, if they exist.

\section{Summary and Discussion}\label{sec: discussion}

Previous studies have examined the Poincaré recurrence theorem in quantum systems, where recurrence is understood as the ability of a state to return close to its initial state, as measured by a suitable distance on Hilbert space. In non-integrable systems, recurrence times typically scale doubly exponentially with the system size, and exponentially in integrable ones \cite{venuti_2015_the}. 
In this work, we focused on exact, state-independent recurrences and presented a method for identifying them based on an arithmetic and number-theoretic study of Floquet unitaries.  Rather than requiring explicit expressions for the corresponding matrix elements, we only needed to identify the number field to which they belong. A key observation is that the structure of the solution depends on the dimension $d$ of the system, encapsulated in the condition that $\phi(n)$ divides $[K:\mathbb{Q}]d!$. 
In models where increasing $d$ corresponds to a classical limit,
this upper bound on possible recurrence times diverges. However, this does not contradict the quantum Poincaré recurrence theorem, which addresses approximate and state-dependent recurrences.

We applied this method to a well-known $\delta$-kicked model: the quantum kicked top, an angular momentum system. For a given spin (angular momentum) value, the Hamiltonian is parameterized by two dimensionless quantities, $\alpha$ and $\kappa$. We investigated exact recurrences in the quantum kicked top for $\alpha = \pi/2$ across all values of $\kappa$ for $j = 3/2$. We showed that, despite the Hilbert space being 4-dimensional, there are many values of $\kappa$ for which the system is neither quasi-periodic nor periodic. 
Using a similar method, we also studied the kicked top for $j = 2$, with $\alpha = \pi/2$ and various values of $\kappa$. We studied 500 different values of $\kappa$ of the form $\kappa=\frac{p}{q} \pi \in [0,4j\pi]$ and report that no additional recurrences were found beyond those previously identified in \cite{anand_davis_ghose_2024}.
Moreover, previous studies have argued that the rationality of kicked top parameters is sufficient to ensure exact recurrences \cite{Ruebeck_Pattanayak_2017}. In contrast, our earlier work \cite{anand_davis_ghose_2024} gave numerical evidence that this condition is in fact insufficient. Here we strengthened that statement to a rigorous proof, showing that even in a deeply quantum regime (i.e.\ a low-spin kicked top) with rational Hamiltonian parameters, exact recurrences can be impossible to realize.  This highlights the nontrivial and highly constrained nature of quantum recurrences and how abstract number-theoretic properties of a quantum mechanical system can manifest into global, state-independent statements about its dynamics.  Interesting future work could be to determine precisely what physical or mathematical property of a Floquet unitary guarantees an exact recurrence.

The method developed here can in principle be applied to study the periodic behavior of any finite-dimensional Floquet quantum system. While the complexity of the analysis increases with the number of free parameters, most Floquet systems used to investigate non-equilibrium phenomena involve a single tunable parameter that influences the dynamics, with other parameters held fixed, making our method an efficient tool for analyzing such systems. 
Although our study is not focused on quantum chaos, another important implication is that the presence of exact recurrences effectively rules out chaotic behavior for the corresponding parameter values.
Our framework can act as a clear and novel probe into testing the absence of chaos.
Understanding these periodic structures—especially in low-dimensional systems—can help in designing more effective experiments to probe quantum chaos. 

The exact recurrences studied in this work apply only to periodically driven \textit{closed} systems governed by unitary dynamics. In general, recurrence theorems are typically studied for conservative systems, both classically and quantum-mechanically. 
The presence of external noise or dissipation renders the system open, leading to non-unitary evolution and mixed quantum states. As a result, information is lost and exact recurrences cannot be expected to survive in general.
Nevertheless, one may still ask how close a system can return to its initial state in the presence of noise or dissipation \cite{liu_recurrence_open_quantum_2024}. 
Future work addressing this question in a state-independent setting would require applying our algebraic framework to the theory of quantum channels.  This would entail finding a suitable definition for the periodicity or quasi-periodicity of a channel.  Another possibility in the context of certain structured noise models is to restrict attention to a decoherence-free subspace (an invariant sector of Hilbert space where the noise acts trivially; see also \cite{Bartlett_superselection_2007}) and investigate any effective unitary dynamics therein.

Finally, the presence of exact quantum recurrences in periodically driven systems can play an important role in time-sensitive quantum technologies. These recurrences provide well-defined time scales that 
serve as natural synchronizing points for quantum sensing protocols, and can enhance precision in parameter estimation tasks. Such dynamics have been proposed for achieving Heisenberg-limited sensitivity~\cite{zou_2025_enhancing,biswas_2025_the} where one can ignore the intermediate dynamics and focus solely on the recurrence times. The ability to predict the exact recurrence times and preserve initial coherence can also be useful in developing different quantum thermodynamic devices such as Floquet-based thermal machines~\cite{ronnie_thrermo_2013,scopa_lindbladfloquet_2018}. Our results contribute to this line of research by identifying parameter regimes where coherent recurrence phenomena persist even in non-integrable settings.

\appendix
\section{Field-Theoretic and Cyclotomic Preliminaries}\label{sec: appendix}

In this section, we introduce the basic notions of fields, field extensions, field degree and some of their properties. We will define all the field theoretic notions invoked in the paper and explain some examples. In this appendix, we will also state the main facts we use in the paper with some standard references for proofs. 

Informally, a field is a set in which addition, (commutative) multiplication, subtraction and a division by any non-zero element are possible. If $L$ and $K$ are two fields, and $K \subset L$ then we say $L$ is a field extension (or simply an extension) of $K$ or equivalently we say that $K$ is a sub-field of $L$. In this paper, we only ever need sub-fields of $\mathbb{C}$, the field of complex numbers. In particular, the characteristic of all the fields we deal with in the paper is $0.$

\subsection{Degree of a field extension and the tower law}

We assume that the reader is familiar with the notion of vector space over an arbitrary field. If $L \supset K$ is a field extension then $L$ is naturally a vector space over the field $K$. Therefore, we can speak of the dimension $dim_{K}(L)$ of $L$ over $K$, which can either be finite or infinite, is usually denoted by $[L:K]$ and called the \textit{degree} of $L$ over $K$. If the degree $[L:K]$ is finite, we say that $L$ is a finite extension of $K$. If $L \supset M \supset K$, where $L\supset K$ is a finite extension then $L \supset M$ and $M \supset K$ are finite extensions and we have the tower law  $$[L:K] = [L:M] [M:K].$$

\subsection{Algebraic extensions, simple extension and primitive elements}\label{A1:algebraic externsions}

Given a field extension $L$ over $K$ we say that an element $\alpha \in L$ is \textit{algebraic} over $K$ or an \textit{algebraic number} over $K$, if it satisfies a polynomial $$\alpha^n + a_1{\alpha}^{n-1} + \ldots + a_n = 0,$$
where the coefficients $a_i$ are in the field $K$. Note that any $a \in K$ is algebraic over $K$ since it satisfies the polynomial $x - a$.   If every element of $L$ is algebraic over $K$ then we say that $L$ is an \textit{algebraic} extension of $K$. A finite extension is always an algebraic extension since if say $L\supset K$ is finite dimensional then for any $\alpha \in L$ there exists an $n \in \mathbb{{N}}$ such that $1, \alpha, \ldots, \alpha^n$ are linearly dependent over $K$, i.e., $a_n \alpha^n + \ldots + a_0 = 0$ for some $a_0, \ldots, a_n \in K,$ not all $0$. Therefore, $\alpha$ is algebraic over $K$.  On the other hand, there exists infinite extensions which are algebraic, for instance the fields $L = \mathbb{Q}(S) \supset K = \mathbb{Q}$, where $S$ is the set of \textit{all roots of unity}, i.e., all $a \in \mathbb{C}$ such that $a^n = 1$ for some $n\geq 1.$ The meaning of `$\mathbb{Q}(S)$' is the smallest sub-field of $\mathbb{C}$ containing field $\mathbb{Q}$ and the set $S.$

\begin{example}
    Consider the field extension $\mathbb{C} \supset \mathbb{R}$. The element $i \in \mathbb{C}$ is algebraic over $\mathbb{R}$ since it satisfies $i^2 + 1 = 0.$ Moreover, any $\alpha = a + i b \in \mathbb{C}$ with $a,b \in \mathbb{R},$ satisfies $(\alpha - a)^2 + b^2 = 0,$ a polynomial of degree $2$ with coefficients in $\mathbb{R}.$
\end{example}

Suppose $L \supset K$ is a field extension and $\alpha \in L$ an algebraic number over $K$. We define $K(\alpha)$ as the smallest subfield of $L$ containing both $K$ and $\alpha$, i.e., $K(\alpha)$ contains $K$ and $\alpha$, and if $M \supset K$ is another subfield of $L$ containing $K$ and $\alpha,$ then $M \supset K(\alpha).$ One can describe $K(\alpha)$ also as $$K(\alpha) = \{ \sum_{i=0}^{N} c_{i} \alpha^{i} : c_{i} \in K\}$$
the set of all polynomials in $\alpha.$ One can prove that every non-zero element of the form $f(\alpha) = \sum_{i=0}^{N} c_{i} \alpha^{i} $ is, in fact, invertible, i.e., there exists $g(\alpha) = \sum_{i=0}^{N} c_{i}' \alpha^{i}$ such that $f(\alpha)g(\alpha) = 1.$ Therefore, $K(\alpha)$ is a field.

Indeed, if $f(\alpha) \neq 0$ then $f(x)$ and the minimal polynomial $p_{\alpha}(x)$ of $\alpha$ are co-prime. This is because the minimal polynomial is irreducible over $K.$ Therefore, there exists $g(x), h(x) \in K[x]$ such that $f(x)g(x) + h(x)p_{\alpha}(x) = 1, i.e.,$ $$f(x)g(x) \equiv 1 \pmod{p_{\alpha}(x)}.$$
It follows that $f(\alpha)g(\alpha) = 1.$

The following is a well-known theorem, 
\begin{theorem} Let $K(\alpha)\supset K$ be a finite extension, or equivalently $\alpha$ is algebraic over $K$ then 
    $$[K(\alpha):K] = deg(p_{\alpha}(x)),$$
   where $p_{\alpha}(x) \in K[x]$ is the minimal polynomial of $\alpha.$ 
\end{theorem}

\subsection{Roots of unity and Cyclotomic fields}

In this subsection, we quickly recall some facts about the cyclotomic fields, i.e., the fields $\mathbb{Q}(\zeta_n)$ generated by primitive $n$-th roots of unity $\zeta_n$ for integers $n \geq 1.$ Again, detailed proofs of the facts stated here can be found in \cite{dummit_foote_abstract_2003} and \cite{lang_2002_algebra}.

 An $n$-th root of unity is a complex number $\alpha$ satisfying $\alpha^n = 1.$ We say that an $n$-th root of unity $\alpha$ is a primitive $n$-th root of unity if $\alpha^k \neq 1$ for $1 \leq k < n.$ There are precisely $\varphi(n)$ primitive $n$-th roots of unity. By definition, the Euler totient function $\varphi(n)$ counts the number of positive integers less than $n$ which are coprime to $n$. This follows from the fact that if $\zeta_n$ denotes a primitive $n$-th root of unity, then for any $k$ relatively prime to $n,$ $\zeta_n^{k}$ is also a primitive $n$-th root of unity. 

 Now let us recall that the $n$-th cyclotomic polynomial $$\varPhi_n(x) := \prod_{1 \leq k < n : gcd(k,n)=1} (x - \zeta_{n}^{k})$$
has \textit{integer coefficients} and that it is the minimal polynomial of $\zeta_n$, therefore irreducible over $\mathbb{Z}$. Note that $deg(\varPhi_n(x)) = \varphi(n)$. One can algorithmically obtain $\varPhi_n(x)$ by the inductive use of the formula $x^n - 1 = \prod_{d |n} \varPhi_d(x).$ 

 \begin{example}
     \begin{itemize}
         \item[1.] $n=4$, $\zeta_4 = i$ and $\varPhi_{4}(x) = x^2 +1.$
         \item [2.] $n=p,$ a prime then $\varPhi_{p}(x) = x^{p-1} + \ldots + x + 1.$
         \item[3.] $n=9,$ $\varPhi_9(x) = x^6 + x^3 + 1.$ 
     \end{itemize}
 \end{example}

The fact that $\varPhi_n(x)$ is the minimal polynomial of $\zeta_n$
implies the following result, $[\mathbb{Q}(\zeta_n) : \mathbb{Q}] = \varphi(n).$ For detailed proof of these results, see Theorem 3.1 of \cite{lang_2002_algebra}.

\subsection{Wigner D matrix}

The rotation matrix in Eq.\eqref{eq:floquet_unitary QKT} belongs to  $SU(r)$, which is here parameterized by the Euler angles ($\theta,\phi,\varphi)$ according to the z-y-z convention. That is, an arbitrary rotation is defined as 
\begin{equation}\label{eq:general rotation}
    \mathcal{R}(\theta,\alpha,\varphi) = e^{-i\theta J_z} e^{-i\phi J_y} e^{-i\varphi J_z}
\end{equation}
where $\theta \in [0,2\pi], \alpha \in \mathbb{R} \, \text{mod} \, 2\pi$ and $\varphi \in \mathbb{R}  \, \text{mod}\,  2\pi$.  When $\mathcal{R}$ is expressed in the Dicke basis it is referred to as the \textit{Wigner D-matrix} \cite{Biedenharn_Louck_Carruthers_1984}. The matrix elements are given by
\begin{equation}
\begin{aligned}
     D^j_{m^{\prime}m}  &= \bra{j,m^{\prime}}\mathcal{R}(\theta,\alpha,\varphi) \ket{j,m}\\
    &=e^{-im^{\prime}\theta}d^j_{{m^{\prime}m}}e^{-im\varphi},
\end{aligned}
\end{equation}
where 
\begin{equation*}
    d^j_{m^{\prime}m}=[\eta(j,m,m')]^{\frac{1}{2}}\sum_{s=s_{min}}^{s=s_{max}}\Bigg[\frac{(-1)^{m^{\prime}-m+s}\big(\cos\frac{\alpha}{2}\big)^{2j+m-m^{\prime}-2s}\big(\sin\frac{\alpha}{2}\big)^{m^{\prime}-m+2s}}{((j+m-s)!s!(m^{\prime}-m+s)!(j-m^{\prime}-s)!}\Bigg],
\end{equation*}
and $\eta(j,m,m') = (j+m^{\prime})!(j-m^{\prime})!(j+m)!(j+m)!$. Here $s_{min}=\text{max}(0,m-m^{prime})$ and $s_{max}=\text{min}(j+m,j-m)$. For the rotation part of the Floquet operator defined in Eq.\eqref{eq:floquet_unitary QKT}, we have $\theta=\varphi=0$. 
Therefore, $D^j_{m^{\prime}m} = d^j_{m^{\prime}m}$. When $\alpha = \pi/2$ we have, $\cos{\alpha/2} = \sin{\alpha/2} = \frac{1}{\sqrt{2}}.$ Therefore, $d^j_{m^{\prime}m}$ is $[\eta(j,m,m')]^{\frac{1}{2}}$ times a rational number. This implies that the rotation matrix has algebraic entries. This remains true when $\alpha$ is any rational multiple of $\pi$, as in such case $\cos{\alpha/2}$ and $\sin{\alpha/2}$ remain algebraic. 

\printbibliography

@article{marklof_arithmetic_2006,
  title={Arithmetic quantum chaos},
  author={Marklof, Jens},
  journal={Encyclopedia of Mathematical Physics},
  volume={1},
  pages={212--220},
  year={2006},
  url={https://doi.org/10.1016/B0-12-512666-2/00449-1},
  publisher={Academic Press Oxford}
}

@article{bogomolny_arithmetical_1997,
  title={Arithmetical chaos},
  author={Bogomolny, Eugene B and Georgeot, Bertrand and Giannoni, M-J and Schmit, Charles},
  journal={Physics Reports},
  volume={291},
  number={5-6},
  pages={219--324},
  year={1997},
  url={https://doi.org/10.1016/S0370-1573(97)00016-1},
  publisher={Elsevier}
}

@article{jens_aqc_1993,
author = {Bolte, Jens},
title = {Some studies on arithmetical chaos in classical and qauntum mechanics },
journal = {International Journal of Modern Physics B},
volume = {07},
number = {27},
pages = {4451-4553},
year = {1993},
URL = {  https://doi.org/10.1142/S0217979293003759}
}

@article{grifoni_driven_1998,
  author = {Grifoni, Milena and Hänggi, Peter},
  month = {10},
  pages = {229-354},
  title = {Driven quantum tunneling},
  url = {https://doi.org/10.1016/S0370-1573(98)00022-2},
  volume = {304},
  year = {1998},
  journal = {Physics Reports}
}

@article{gfloquet_1883,
  author = {G. Floquet},
  month = {01},
  pages = {47-88},
  publisher = {Société Mathématique de France},
  title = {Sur les équations différentielles linéaires à coefficients périodiques},
  url = {https://doi.org/10.24033/asens.220},
  volume = {12},
  year = {1883},
  journal = {Annales Scientifiques de l'École Normale Supérieure}
}

@book{dummit_foote_abstract_2003, 
    address={New York}, 
    edition={3}, 
    title={{Abstract Algebra}}, 
    ISBN={978-0-471-43334-7}, 
    publisher={John Wiley \& Sons}, 
    author={Dummit, David S. and Foote, Richard M.}, 
    year={2003} 
}

@article{Baguette_MM_reductions_2014,
  title = {Multiqubit symmetric states with maximally mixed one-qubit reductions},
  author = {Baguette, D. and Bastin, T. and Martin, J.},
  journal = {Physical Review A},
  volume = {90},
  issue = {3},
  pages = {032314},
  numpages = {10},
  year = {2014},
  month = {Sep},
  publisher = {American Physical Society},
  url = {https://doi.org/10.1103/PhysRevA.90.032314}
}

@article{kumari_2019_quantumclassical,
  author = {Kumari, Meenu},
  month = {07},
  publisher = {University of Waterloo},
  title = {Quantum-Classical Correspondence and Entanglement in Periodically Driven Spin Systems},
  url = {http://doi.org/10012/14860},
  year = {2019},
  journal = {University of Waterloo}
}

@article{haake_1987,
  title={Classical and quantum chaos for a kicked top},
  author={Fritz Haake and Marek Kuś and Rainer Scharf},
  journal={Zeitschrift f{\"u}r Physik B Condensed Matter},
  year={1987},
  volume={65},
  pages={381-395},
  doi={https://doi.org/10.1007/BF01303727}
}

@article{harrow2013church,
  title={The church of the symmetric subspace},
  author={Harrow, Aram W},
  journal={arXiv:1308.6595},
  year={2013},
  url={https://doi.org/10.48550/arXiv.1308.6595}
}

@book{Biedenharn_Louck_Carruthers_1984, 
    edition={1}, 
    title={{Angular Momentum in Quantum Physics: Theory and Application}}, 
    rights={https://www.cambridge.org/core/terms}, 
    ISBN={978-0-521-30228-9}, 
    url={https://doi.org/10.1017/CBO9780511759888}, 
    publisher={Cambridge University Press}, 
    author={Biedenharn, L. C. and Louck, James D. and Carruthers, Peter A.}, 
    year={1984}
}

@article{fishman_rotor_anderson_1982,
  title = {Chaos, Quantum Recurrences, and {Anderson} Localization},
  author = {Fishman, Shmuel and Grempel, D. R. and Prange, R. E.},
  journal = {Physical Review Letters},
  volume = {49},
  issue = {8},
  pages = {509--512},
  numpages = {0},
  year = {1982},
  month = {Aug},
  publisher = {American Physical Society},
  url = {https://doi.org/10.1103/PhysRevLett.49.509}
}

@article{Izrailev_Shepelyanskii_1980, 
    title={Quantum resonance for a rotator in a nonlinear periodic field}, 
    volume={43}, 
    ISSN={0040-5779, 1573-9333}, 
    url ={https://doi.org/10.1007/BF01029131}, 
    number={3}, 
    journal={Theoretical and Mathematical Physics}, 
    author={Izrailev, F. M. and Shepelyanskii, D. L.}, 
    year={1980}, 
    month={Jun}, 
    pages={553–561}, 
}

@article{Zou_Wang_pseudo_2022, 
    author={Zou, Zhixing and Wang, Jiao}, 
    title={Pseudoclassical Dynamics of the Kicked Top}, 
    volume={24}, 
    ISSN={1099-4300}, 
    url ={https://doi.org/10.3390/e24081092}, 
    number={8}, 
    journal={Entropy}, 
    year={2022}, 
    month={Aug}, 
    pages={1092}, 
}

@article{Bhosale_Santhanam_periodicity_2018,
  title = {Periodicity of quantum correlations in the quantum kicked top},
  author = {Bhosale, Udaysinh T. and Santhanam, M. S.},
  journal = {Physical Review E},
  volume = {98},
  issue = {5},
  pages = {052228},
  numpages = {11},
  year = {2018},
  month = {Nov},
  publisher = {American Physical Society},
  url = {https://doi.org/10.1103/PhysRevE.98.052228}
}

@article{Ruebeck_Pattanayak_2017,
  title = {Entanglement and its relationship to classical dynamics},
  author = {Ruebeck, Joshua B. and Lin, Jie and Pattanayak, Arjendu K.},
  journal = {Physical Review E},
  volume = {95},
  issue = {6},
  pages = {062222},
  numpages = {10},
  year = {2017},
  month = {Jun},
  publisher = {American Physical Society},
  url = {https://doi.org/10.1103/PhysRevE.95.062222}
}

@article{Bartlett_superselection_2007,
  title = {Reference frames, superselection rules, and quantum information},
  author = {Bartlett, Stephen D. and Rudolph, Terry and Spekkens, Robert W.},
  journal = {Rev. Mod. Phys.},
  volume = {79},
  issue = {2},
  pages = {555--609},
  numpages = {0},
  year = {2007},
  month = {Apr},
  publisher = {American Physical Society},
  url = {https://doi.org/10.1103/RevModPhys.79.555}
}

@article{Dogra_exactly_2019,
  title = {Quantum signatures of chaos, thermalization, and tunneling in the exactly solvable few-body kicked top},
  author = {Dogra, Shruti and Madhok, Vaibhav and Lakshminarayan, Arul},
  journal = {Physical Review E},
  volume = {99},
  issue = {6},
  pages = {062217},
  numpages = {15},
  year = {2019},
  month = {Jun},
  publisher = {American Physical Society},
  url = {https://doi.org/10.1103/PhysRevE.99.062217}
}

@article{kaufman_2016_quantum,
  author = {Kaufman, Adam and M. Eric Tai and Lukin, Alexander and Rispoli, Matthew and Schittko, Robert and Preiss, Philipp M and Greiner, Markus},
  month = {08},
  pages = {794-800},
  publisher = {American Association for the Advancement of Science},
  title = {Quantum thermalization through entanglement in an isolated many-body system},
  url = {https://doi.org/10.1126/science.aaf6725},
  volume = {353},
  year = {2016},
  journal = {Science}
}

@article{roy_2017_floquet,
  author = {Roy, Rahul and Harper, Fenner},
  month = {05},
  publisher = {American Physical Society},
  title = {{Floquet} topological phases with symmetry in all dimensions},
  url = {https://doi.org/10.1103/PhysRevB.95.195128},
  volume = {95},
  year = {2017},
  journal = {Physical Review B}
}

@article{sacha_2017_time,
  author = {Sacha, Krzysztof and Zakrzewski, Jakub},
  month = {11},
  pages = {016401},
  title = {Time crystals: a review},
  url = {https://doi.org/10.1088/1361-6633/aa8b38},
  volume = {81},
  year = {2017},
  journal = {Reports on Progress in Physics} 
}

@article{anand_davis_ghose_2024,
  title = {Quantum recurrences in the kicked top},
  author = {Anand, Amit and Davis, Jack and Ghose, Shohini},
  journal = {Physical Review Research},
  volume = {6},
  issue = {2},
  pages = {023120},
  numpages = {9},
  year = {2024},
  month = {May},
  publisher = {American Physical Society},
  url = {https://doi.org/10.1103/PhysRevResearch.6.023120}
}

@article{khemani_2019_a,
  author = {Khemani, Vedika and Moessner, Roderich and Sondhi, S. L.},
  month = {10},
  title = {A Brief History of Time Crystals},
  url = {https://doi.org/10.48550/arXiv.1910.10745},
  year = {2019},
  organization = {arXiv.org},
  journal = {arXiv:1910.10745 }
}

@article{Arecchi1972SCS,
  title = {Atomic Coherent States in Quantum Optics},
  author = {Arecchi, F. T. and Courtens, Eric and Gilmore, Robert and Thomas, Harry},
  journal = {Physical Review A},
  volume = {6},
  issue = {6},
  pages = {2211--2237},
  numpages = {0},
  year = {1972},
  month = {Dec},
  publisher = {American Physical Society},
  url = {https://doi.org/10.1103/PhysRevA.6.2211}
}

@article{hogg_1982_recurrence,
  author = {Hogg, T. and Huberman, B. A.},
  month = {03},
  pages = {711-714},
  title = {Recurrence Phenomena in Quantum Dynamics},
  url = {https://doi.org/10.1103/PhysRevLett.48.711},
  volume = {48},
  year = {1982},
  journal = {Physical Review Letters}
}

@article{Norio_scars_2020,
  title = {Exact {Floquet} quantum many-body scars under {Rydberg} blockade},
  author = {Mizuta, Kaoru and Takasan, Kazuaki and Kawakami, Norio},
  journal = {Physical Review Research},
  volume = {2},
  issue = {3},
  pages = {033284},
  numpages = {13},
  year = {2020},
  month = {Aug},
  publisher = {American Physical Society},
  url = {https://doi.org/10.1103/PhysRevResearch.2.033284}
}

@article{Pai_and_Pretko_scar_2019,
  title = {Dynamical Scar States in Driven Fracton Systems},
  author = {Pai, Shriya and Pretko, Michael},
  journal = {Physical Review Letters},
  volume = {123},
  issue = {13},
  pages = {136401},
  numpages = {5},
  year = {2019},
  month = {Sep},
  publisher = {American Physical Society},
  url = {https://doi.org/10.1103/PhysRevLett.123.136401}
}

@article{Sen_scar_2020,
  title = {Collapse and revival of quantum many-body scars via {Floquet} engineering},
  author = {Mukherjee, Bhaskar and Nandy, Sourav and Sen, Arnab and Sen, Diptiman and Sengupta, K.},
  journal = {Physical Review B},
  volume = {101},
  issue = {24},
  pages = {245107},
  numpages = {12},
  year = {2020},
  month = {Jun},
  publisher = {American Physical Society},
  url = {https://doi.org/10.1103/PhysRevB.101.245107}
}

@book{lang_2002_algebra, 
    address={New York, NY}, 
    edition={Revised Third Edition}, 
    series={Graduate Texts in Mathematics}, 
    title={Algebra}, 
    ISBN={978-1-4613-0041-0}, 
    DOI={10.1007/978-1-4613-0041-0}, 
    publisher={Springer}, 
    author={Lang, Serge}, 
    year={2002}, 
    collection={Graduate Texts in Mathematics}
}

@article{biswas_2025_the,
  author = {Biswas, Hillol and Choudhury, Sayan},
  title = {The {Floquet} central spin model: A platform to realize eternal time crystals, entanglement steering, and multiparameter metrology},
  url = {https://doi.org/10.48550/arXiv.2501.18472},
  year = {2025},
  organization = {arXiv.org},
  journal = {arXiv:2501.18472 }
}

@article{poincar_1890_avantpropos,
author = {Henri Poincaré},
title = {{Sur le problème des trois corps et les équations de la dynamique}},
volume = {13},
journal = {Acta Mathematica},
number = {1-2},
publisher = {Institut Mittag-Leffler},
pages = {VII},
year = {1890},
URL = {https://doi.org/10.1007/BF02392505}
}

@article{bocchieri_1957_quantum,
  author = {Bocchieri, P and A. Loinger},
  month = {07},
  pages = {337-338},
  publisher = {American Institute of Physics},
  title = {Quantum Recurrence Theorem},
  url = {https://doi.org/10.1103/PhysRev.107.337},
  volume = {107},
  year = {1957},
  journal = {Physical Review}
}

@article{pandit_2022_bounds,
  author = {Pandit, Tanmoy and Green, Alaina M and Alderete, C Huerta and Linke, Norbert M and Raam Uzdin},
  month = {04},
  pages = {682-682},
  publisher = {Verein zur Förderung des Open Access Publizierens in den Quantenwissenschaften},
  title = {Bounds on the recurrence probability in periodically-driven quantum systems},
  url = {https://doi.org/10.22331/q-2022-04-06-682},
  volume = {6},
  year = {2022},
  journal = {Quantum}
}

@article{oszmaniec_2024_saturation,
  author = {Oszmaniec, Michał and Kotowski, Marcin and Horodecki, Michał and Hunter-Jones, Nicholas},
  month = {12},
  publisher = {American Physical Society (APS)},
  title = {Saturation and Recurrence of Quantum Complexity in Random Local Quantum Dynamics},
  url = {https://doi.org/10.1103/PhysRevX.14.041068},
  volume = {14},
  year = {2024},
  journal = {Physical Review X}
}

@article{3qubits2ways_Dur_2000,
  title = {Three qubits can be entangled in two inequivalent ways},
  author = {D\"ur, W. and Vidal, G. and Cirac, J. I.},
  journal = {Physical Review A},
  volume = {62},
  issue = {6},
  pages = {062314},
  numpages = {12},
  year = {2000},
  month = {Nov},
  publisher = {American Physical Society},
  url = {https://doi.org/10.1103/PhysRevA.62.062314}
}

@article{venuti_2015_the,
  author = {Venuti, Lorenzo Campos},
  month = {01},
  publisher = {},
  title = {The recurrence time in quantum mechanics},
  url = {https://doi.org/10.48550/arXiv.1509.04352},
  year = {2015},
  journal = {arXiv:1509.04352 }
}

@article{Giraud_Tensor_PRL_2015,
  title = {Tensor Representation of Spin States},
  author = {Giraud, O. and Braun, D. and Baguette, D. and Bastin, T. and Martin, J.},
  journal = {Physical Review Letters},
  volume = {114},
  issue = {8},
  pages = {080401},
  numpages = {5},
  year = {2015},
  month = {Feb},
  publisher = {American Physical Society},
  url = {https://doi.org/10.1103/PhysRevLett.114.080401}
}

@article{zou_2025_enhancing,
  author = {Zou, Zhixing and Gong, Jiangbin and Chen, Weitao},
  month = {05},
  publisher = {American Physical Society},
  title = {Enhancing quantum metrology by quantum resonance dynamics},
  url = {https://doi.org/10.1103/lkrt-lvng},
  volume = {134},
  year = {2025},
  journal = {Physical Review Letters}
}

@article{bukov_many_body_2021,
  title = {Prethermalization and thermalization in periodically driven many-body systems away from the high-frequency limit},
  author = {Fleckenstein, Christoph and Bukov, Marin},
  journal = {Physical Review B},
  volume = {103},
  issue = {14},
  pages = {L140302},
  numpages = {6},
  year = {2021},
  month = {Apr},
  publisher = {American Physical Society},
  url = {https://doi.org/10.1103/PhysRevB.103.L140302}
}

@article{levesque_rec_time_2025,
  author = {Levesque, Dominique and Sourlas, Nicolas},
  month = {06},
  publisher = {Springer Science+Business Media},
  title = {Time Irreversibility in Statistical Mechanics},
  url = {https://doi.org/10.1007/s10955-025-03467-0},
  volume = {192},
  year = {2025},
  journal = {Journal of Statistical Physics}
}

@article{suskind_complexity_2018,
  title = {Second law of quantum complexity},
  author = {Brown, Adam R. and Susskind, Leonard},
  journal = {Physical Review D},
  volume = {97},
  issue = {8},
  pages = {086015},
  numpages = {29},
  year = {2018},
  month = {Apr},
  publisher = {American Physical Society},
  url = {https://doi.org/10.1103/PhysRevD.97.086015}
}

@article{gimeno_upper_2017,
  title = {Upper bounds for the {Poincaré} recurrence time in quantum mixed states},
  author = {Gimeno, V and Sotoca, J M},
  journal = {Journal of Physics A Mathematical and Theoretical},
  month = {03},
  publisher = {Institute of Physics},
  volume = {50},
  year = {2017},
  url = {https://doi.org/10.1088/1751-8121/aa67fe}
}

@article{ropotenko_the_2025,
  author = {Ropotenko, K},
  title = {The {Poincar\'e} recurrence time for the de {Sitter} space with dynamical chaos},
  url = {https://doi.org/10.48550/arXiv.0712.0993},
  year = {2025},
  journal = {arXiv:0712.0993 }
}

@article{bhosle_qkt__jan_2024,
  title = {Exactly solvable dynamics and signatures of integrability in an infinite-range many-body {Floquet} spin system},
  author = {Sharma, Harshit and Bhosale, Udaysinh T.},
  journal = {Physical Review B},
  volume = {109},
  issue = {1},
  pages = {014412},
  numpages = {8},
  year = {2024},
  month = {Jan},
  publisher = {American Physical Society},
  url = {https://doi.org/10.1103/PhysRevB.109.014412}
}

@article{bhosle_qkt_aug_2024,
  title = {Signatures of quantum integrability and exactly solvable dynamics in an infinite-range many-body {Floquet} spin system},
  author = {Sharma, Harshit and Bhosale, Udaysinh T.},
  journal = {Physical Review B},
  volume = {110},
  issue = {6},
  pages = {064313},
  numpages = {14},
  year = {2024},
  month = {Aug},
  publisher = {American Physical Society},
  url = {https://doi.org/10.1103/PhysRevB.110.064313}
}

@article{bhosle_nov_2024,
      title={Exact Solvability Of Entanglement For Arbitrary Initial State in an Infinite-Range {Floquet} System}, 
      author={Harshit Sharma and Udaysinh T. Bhosale},
      journal = {Annals of Physics},
      volume = {486},
      pages = {170327},
      year = {2026},
      issn = {0003-4916},
      url={https://doi.org/10.1016/j.aop.2025.170327}
}

@article{alvaro_recurrence_2023,
	title={{Concentration of quantum equilibration and an estimate of the recurrence time}},
	author={Jonathon Riddell and Nathan J. Pagliaroli and Álvaro M. Alhambra},
	journal={SciPost Phys.},
	volume={15},
	pages={165},
	year={2023},
	publisher={SciPost},
	url={https://doi.org/10.21468/SciPostPhys.15.4.165},
}

@article{scopa_lindbladfloquet_2018,
  title = {{Lindblad-Floquet} description of finite-time quantum heat engines},
  author = {Scopa, Stefano and Landi, Gabriel T. and Karevski, Dragi},
  journal = {Physical Review A},
  volume = {97},
  issue = {6},
  pages = {062121},
  numpages = {12},
  year = {2018},
  month = {Jun},
  publisher = {American Physical Society},
  url = {https://doi.org/10.1103/PhysRevA.97.062121}
}

@article{ronnie_thrermo_2013,
AUTHOR = {Kosloff, Ronnie},
TITLE = {Quantum Thermodynamics: A Dynamical Viewpoint},
JOURNAL = {Entropy},
VOLUME = {15},
YEAR = {2013},
NUMBER = {6},
PAGES = {2100--2128},
URL = {https://doi.org/10.3390/e15062100},
ISSN = {1099-4300}
}

@article{freedman_quantumrecurrent_2024,
      title={Quantum Detection of Recurrent Dynamics}, 
      author={Michael H. Freedman},
      year={2024},
      eprint={2407.16055},
      archivePrefix={arXiv 2407.16055},
      primaryClass={quant-ph},
      url={https://doi.org/10.48550/arXiv.2407.16055},
      JOURNAL = {arXiv:2407.16055},
}

@article{kotowski_tightboundsrecurrencetime_2026,
      title={Tight bounds on recurrence time in closed quantum systems}, 
      author={Marcin Kotowski and Michał Oszmaniec},
      year={2026},
      eprint={2601.10409},
      archivePrefix={arXiv},
      primaryClass={quant-ph},
      url={https://doi.org/10.48550/arXiv.2601.10409}, 
      JOURNAL = {arXiv:2601.10409},
}

@article{rauer_recurrence_science_2018,
author = {Bernhard Rauer  and Sebastian Erne  and Thomas Schweigler  and Federica Cataldini  and Mohammadamin Tajik  and Jörg Schmiedmayer },
title = {Recurrences in an isolated quantum many-body system},
journal = {Science},
volume = {360},
number = {6386},
pages = {307-310},
year = {2018},
URL = {https://doi.org/10.1126/science.aan7938},
eprint = {https://www.science.org/doi/pdf/10.1126/science.aan7938}}

@article{liu_recurrence_open_quantum_2024,
      title={Recurrence Theorem for Open Quantum Systems}, 
      author={Zhihang Liu and Chao Zheng},
      year={2024},
      eprint={2402.19143},
      archivePrefix={arXiv},
      primaryClass={quant-ph},
      url={https://doi.org/10.48550/arXiv.2402.19143}, 
       JOURNAL = {arXiv:2402.19143},
      
}

@article{sayan_time_crystal_2020,
  title = {Eternal discrete time crystal beating the Heisenberg limit},
  author = {Lyu, Changyuan and Choudhury, Sayan and Lv, Chenwei and Yan, Yangqian and Zhou, Qi},
  journal = {Phys. Rev. Res.},
  volume = {2},
  issue = {3},
  pages = {033070},
  numpages = {11},
  year = {2020},
  month = {Jul},
  publisher = {American Physical Society},
  url = {https://doi.org/10.1103/PhysRevResearch.2.033070}
}

\end{document}